\documentclass[review]{elsarticle}

\usepackage{lineno,hyperref,mathtools}
\usepackage[letterpaper,top=2cm,bottom=2cm,left=3cm,right=3cm,marginparwidth=1.75cm]{geometry}

\journal{Journal of \LaTeX\ Templates}

\usepackage{amssymb}
\usepackage{bm}
\usepackage{amssymb}
\usepackage{amsthm}
\usepackage{multirow,booktabs}
\usepackage{cases}
\usepackage{tabularx}
\usepackage{adjustbox}
\usepackage[figuresright]{rotating}
\usepackage{float}

\makeatletter

\newcommand{\Rmnum}[1]{\expandafter\@slowromancap\romannumeral #1@}
\makeatother

\newtheorem{theorem}{Theorem}

\newtheorem{remark}{Remark}

\newtheorem{lemma}[theorem]{Lemma}

\newtheorem{example}[theorem]{Example}
\newcommand{\GRS}{{\mathrm{GRS}}}
\newcommand{\Hull}{{\mathrm{Hull}}}
\newcommand{\C}{{\mathcal{C}}}
\newcommand{\F}{{\mathbb{F}}}
\newcommand{\bc}{{{\bf c}}}

\begin{document}

\begin{frontmatter}

\title{On MDS Codes With Galois Hulls of Arbitrary Dimensions}
\tnotetext[mytitlenote]{This research was supported by the National Natural Science Foundation of China under Grant Nos.U21A20428 and 12171134.}


\author[address1]{Yang Li}
\ead{yanglimath@163.com}

\author[address1]{Shixin Zhu\corref{mycorrespondingauthor}}
\cortext[mycorrespondingauthor]{Corresponding author}
\ead{zhushixinmath@hfut.edu.cn}

\author[address1]{Ping Li}
\ead{lpmath@126.com}

\address[address1]{School of Mathematics, HeFei University of Technology, Hefei 230601, China}

\begin{abstract}
The Galois hulls of linear codes are a generalization of the Euclidean and Hermitian hulls of linear codes. In this paper, we study the Galois hulls of (extended) GRS codes 
and present several new constructions of MDS codes with Galois hulls of arbitrary dimensions via (extended) GRS codes. Two general methods of constructing MDS codes 
with Galois hulls of arbitrary dimensions by Hermitian or general Galois self-orthogonal (extended) GRS codes are given. Using these methods, some MDS codes with larger dimensions 
and Galois hulls of arbitrary dimensions can be obtained and relatively strict conditions can also lead to many new classes of MDS codes with Galois hulls of arbitrary dimensions.
\end{abstract}


\end{frontmatter}


\section{Introduction}\label{sec-introduction}

Throughout this paper, let $q=p^h$, where $p$ is a prime and $h$ is a positive integer. Let $\F_q$ be the finite field with $q$ elements. Let $\bc_1=(c_{11},c_{12},\ldots,c_{1n})$ and 
$\bc_2=(c_{21},c_{22},\ldots,c_{2n})\in \F_q^n$. The $e$-Galois inner product of $\bc_1$ and $\bc_2$ is defined by $(\bc_1,\bc_2)_e=c_{11}c_{21}^{p^e}+c_{21}c_{22}^{p^e}+\cdots+c_{1n}c_{2n}^{p^e}$, 
where $0\leq e\leq h-1$, which was first introduced by Fan et al. \cite[]{RefJ95}. In particular, if $e=0$, the $e$-Galois inner product coincides with the usual Euclidean inner 
product. Moreover, if $h$ is even and $e=h/2$, the $e$-Galois inner product coincides with the Hermitian inner product. 

An $[n,k,d]_q$ linear code $\C$ is a $k$-dimensional subspace of $\F_q^n$ with minimum distance $d$. Its $e$-Galois dual code, denoted by $\C^{\bot_e}$, is defined by 
$\C^{\bot_e}=\{ \bc_1 \in \F_q^n:\ (\bc_1,\bc_2)_e=0,\ {\rm for \ all} \ \bc_2\in \C \}.$ The $e$-Galois hull of the code $\C$ is defined by $\Hull_e(\C)=\C\cap \C^{\bot_e}$. 
If $e=0$, $\Hull_e(\C)$ is called the Euclidean hull of $\C$. If $h$ is even and $e=h/2$, $\Hull_e(\C)$ is called the Hermitian hull of $\C$. Hence, $e$-Galois hulls of linear 
codes are a generalization of the Euclidean and Hermitian hulls of linear codes.

For the Euclidean hull of linear codes, it was first introduced by Assmus \cite{RefJ1} to classify finite projective planes. It had been shown that the hull of linear codes plays an important role 
in many other aspects of coding theory, such as determining the complexity of algorithms for checking the permutation equivalence of two linear codes, computing the automorphism group of a linear code 
and increasing a security level against side channel attacks and fault injection attacks \cite{RefJ100,RefJ8,RefJ9,RefJ10,RefJ11}. They usually work very well when the dimension of the Euclidean hull of a 
linear code is small. 

With the development of quantum computation and quantum communication, how to construct quantum codes to counteract the noise in quantum channels becomes an important and difficult problem in quantum 
information theory. In 1996, an explicit method, called CSS construction, proposed by Calderbank et al. \cite[]{RefJ-1} and Steane \cite[]{RefJ-2}, makes it possible to construct quantum stabilizer 
codes from certain self-orthogonal or dual containing codes. However, the self-orthogonality is often difficult to obtain. About ten years later, Brun et al. \cite{RefJ14} introduced 
entanglement-assisted quantum error-correction codes (EAQECCs). In their constructions, an EAQECC can be derived from any classical linear codes with the help of pre-shared entanglement between the 
encoder and decoder. However, the determination of the number of shared pairs that required is usually difficult. Fortunately, one found certain relationships between this number and the dimension of 
the hull of a linear code. Specifically, we refer to \cite[]{RefJ4} for the usual Euclidean and classical Hermitian cases, and \cite[]{RefJ27} for the general Galois case. Based on these outstanding 
results, a large number of MDS codes (i.e., $d=n-k+1$) with Euclidean hulls and Hermitian hulls of arbitrary dimensions were constructed and so many EAQECCs were given in 
\cite[]{RefJ21,RefJ7,RefJ5,RefJ19,RefJ6,RefJ16,RefJ18,RefJ17} and references therein.

However, up to the authors' knowledge, there seems to be little research on MDS codes with Galois hulls of arbitrary dimensions. About only twenty classes were construted by Qian et al., Cao and Fang et 
al. in \cite[]{RefJ29,RefJ30,RefJ94}. Very recently, Wu et al. \cite[]{RefJ90} studied the Galois hulls of generalized Reed-Solomon (GRS) codes. 
They proved that the Galois hulls of some GRS 
codes are still GRS codes when $\mathfrak{L} =\mathfrak{L}^{p^{h-e}}$ in terms of Goppa codes. In their conditions, the dimensions of Galois hulls are also relatively flexible. Considering 
the excellent properties of the hull of linear codes and its important applications in coding theory, as well as relatively little research under Galois inner product, it is necessary to study 
and construct linear codes with Galois hulls of arbitrary dimensions, especially MDS codes with Galois hulls of arbitrary dimensions. 

In this paper, we study the Galois hulls of (extended) GRS codes and present some new constructions of MDS codes with Galois hulls of arbitrary dimensions via (extended) GRS codes. For reference, 
we list the parameters of all known MDS codes with Galois hulls of arbitrary dimensions in Table \ref{tab:1} and the new ones in Table \ref{tab:2}.

The rest of this paper is organized as follows. Some basic knowledge about (extended) GRS codes are introduced in Section \ref{sec2}. The six new constructions of MDS codes with Galois hulls of 
arbitrary dimensions are discussed in Section \ref{sec3}. And finally, Section \ref{sec4} concludes this paper.

\newcommand{\tabincell}[2]{\begin{tabular}{@{}#1@{}}#2\end{tabular}}
\begin{table}
\centering
\caption{Known constructions on MDS codes with Galois hulls of arbitrary dimensions}
\label{tab:1}    

\begin{center}
  \resizebox{160mm}{100mm}{
	\begin{tabular}{ccccc}
		\hline
	  Class & $q$-Ary & Code length $n$ & Dimension $k$ & Ref.\\
		\hline
   1 & $q=p^h$ is even & $n\leq q$, $\frac{m}{\gcd(e,m)}$ and $m>1$ & $1\leq k\leq \lfloor \frac{p^e+n-1}{p^e+1} \rfloor$ & \cite{RefJ94}  \\

   2 & $q=p^h>3$ & $n\leq r$, $r=p^m$ with $m\mid h$ and $(p^e+1)\mid \frac{q-1}{r-1}$ & $1\leq k\leq \lfloor \frac{p^e+n-1}{p^e+1} \rfloor$ & \cite{RefJ94}  \\

   3 & $q=p^h>3$ & $n\mid q$ & $1\leq k\leq \lfloor \frac{p^e+n-1}{p^e+1} \rfloor$ & \cite{RefJ94}  \\

   4 & $q=p^h>3$ & $(n-1)\mid (q-1)$ & $1\leq k\leq \lfloor \frac{p^e+n-1}{p^e+1} \rfloor$ & \cite{RefJ94}  \\

   5 & $q=p^h>3$ & $n\mid (q-1)$ & $1\leq k\leq \lfloor \frac{p^e+n-1}{p^e+1} \rfloor$ & \cite{RefJ94}  \\

   6 & $q=p^h$ is odd, $2e\mid h$ & \tabincell{c}{$n=\frac{r(q-1)}{\gcd(x_2,q-1)}$, $1\leq r\leq \frac{q-1}{\gcd(x_1,q-1)}$,\\ $(q-1)\mid {\rm lcm}(x_1,x_2)$ and $\frac{q-1}{p^e-1}\mid x_1$} & $1\leq k\leq \lfloor \frac{p^e+n}{p^e+1} \rfloor$ & \cite{RefJ29}  \\

   7 & $q=p^h$ is odd, $2e\mid h$ & \tabincell{c}{$n=\frac{r(q-1)}{\gcd(x_2,q-1)}+1$, $1\leq r\leq \frac{q-1}{\gcd(x_1,q-1)}$,\\ $(q-1)\mid {\rm lcm}(x_1,x_2)$ and $\frac{q-1}{p^e-1}\mid x_1$} & $1\leq k\leq \lfloor \frac{p^e+n}{p^e+1} \rfloor$ & \cite{RefJ29}  \\

   8 & $q=p^h$ is odd, $2e\mid h$ & \tabincell{c}{$n=\frac{r(q-1)}{\gcd(x_2,q-1)}+2$, $1\leq r\leq \frac{q-1}{\gcd(x_1,q-1)}$,\\ $(q-1)\mid {\rm lcm}(x_1,x_2)$ and $\frac{q-1}{p^e-1}\mid x_1$} & $1\leq k\leq \lfloor \frac{p^e+n}{p^e+1} \rfloor$ & \cite{RefJ29}  \\

   9 & $q=p^h$ is odd, $2e\mid h$ & \tabincell{c}{$n=rm$, $1\leq r\leq \frac{p^e-1}{m_1}$,\\ $m_1=\frac{m}{\gcd(m,y)}$, $m\mid (q-1)$ and $y=\frac{q-1}{p^e-1}$} & $1\leq k\leq \lfloor \frac{p^e+n}{p^e+1} \rfloor$ & \cite{RefJ29}  \\
   
   10 & $q=p^h$ is odd, $2e\mid h$ & \tabincell{c}{$n=rm+1$, $1\leq r\leq \frac{p^e-1}{m_1}$,\\ $m_1=\frac{m}{\gcd(m,y)}$, $m\mid (q-1)$ and $y=\frac{q-1}{p^e-1}$} & $1\leq k\leq \lfloor \frac{p^e+n}{p^e+1} \rfloor$ & \cite{RefJ29}  \\

  11& $q=p^h$ is odd, $2e\mid h$ & \tabincell{c}{$n=rm+2$, $1\leq r\leq \frac{p^e-1}{m_1}$,\\ $m_1=\frac{m}{\gcd(m,y)}$, $m\mid (q-1)$ and $y=\frac{q-1}{p^e-1}$} & $1\leq k\leq \lfloor \frac{p^e+n}{p^e+1} \rfloor$ & \cite{RefJ29}  \\

  12 & $q=p^h$ is odd, $2e\mid h$ & $n=tp^{aw}$, $1\leq t\leq p^a$, $1\leq w\leq \frac{h}{a}-1$, $a\mid e$ & $1\leq k\leq \lfloor \frac{p^e+n-1}{p^e+1} \rfloor$ & \cite{RefJ29}\\

  13 & $q=p^h$ is odd, $2e\mid h$ & $n=tp^{aw}+1$, $1\leq t\leq p^a$, $1\leq w\leq \frac{h}{a}-1$, $a\mid e$ & $1\leq k\leq \lfloor \frac{p^e+n-1}{p^e+1} \rfloor$ & \cite{RefJ29}\\

  14 & $q=p^h$ is odd, $2e\mid h$ & $n=\frac{t(q-1)}{p^e-1}$, $1\leq t\leq p^e-1$ & $1\leq k\leq \lfloor \frac{p^e+n}{p^e+1} \rfloor$ & \cite{RefJ29}\\

  15 & $q=p^h$ is odd, $2e\mid h$ & $n=\frac{t(q-1)}{p^e-1}+1$, $1\leq t\leq p^e-1$ & $1\leq k\leq \lfloor \frac{p^e+n}{p^e+1} \rfloor$ & \cite{RefJ29}\\

  16 & $q=p^h$ is odd, $2e\mid h$ & $n=\frac{t(q-1)}{p^e-1}+2$, $1\leq t\leq p^e-1$ & $1\leq k\leq \lfloor \frac{p^e+n}{p^e+1} \rfloor$ & \cite{RefJ29}\\

  17 & $q=p^{em}$ is odd, $m$ is even & $n=tp^{er}$, $t\mid (p^e-1)$, $r\leq m-1$  & $1\leq k\leq \lfloor \frac{p^e+n-1}{p^e+1} \rfloor$ & \cite{RefJ30}\\

  18 & $q=p^{h}$ is odd, $2e\mid h$ & $n=tp^{h-e}$, $1\leq t\leq p^e$ & $1\leq k\leq \lfloor \frac{p^e+n-1}{p^e+1} \rfloor$ & \cite{RefJ30}\\

  19 & $q=p^{h}$ is odd & \tabincell{c}{$1\leq m\leq \lfloor \frac{n}{2}\rfloor$, $\frac{h}{e}$ is odd,\\ $\GRS_m(\boldsymbol{a},\boldsymbol{v})^{\bot_0}=\GRS_{n-m}(\boldsymbol{a},\boldsymbol{v})$}  & $1\leq k\leq \lfloor \frac{p^e+n-1}{p^e+1} \rfloor $ & \cite{RefJ30}\\
  
  20 & $q=p^{h}$ is odd & \tabincell{c}{$1\leq m\leq \lfloor \frac{n+1}{2}\rfloor$, $\frac{h}{e}$ is odd,\\ $\GRS_m(\boldsymbol{a},\boldsymbol{v},\infty)^{\bot_0}=\GRS_{n-m}(\boldsymbol{a},\boldsymbol{v},\infty)
         $} & $1\leq k\leq \lfloor \frac{p^e+n-1}{p^e+1} \rfloor$ & \cite{RefJ30}\\
  \hline
	\end{tabular}}
\end{center}

\end{table}

\begin{table}[!htb]
\centering
\caption{The new constructions of MDS codes with Galois hulls of arbitrary dimensions}
\label{tab:2}    

\begin{center}
  \resizebox{160mm}{45mm}{
	\begin{tabular}{ccccc}
		\hline
	  Class & $q$-Ary & Code length $n$ & Dimension $k$ & Ref.\\
		\hline
        1 &  $q=p^{h}\geq 5$  & \tabincell{c}{$\gcd(e',h)=e$, $\frac{h}{e}$ is even,\\ $\GRS_m(\boldsymbol{a},\boldsymbol{v})\subseteq \GRS_m(\boldsymbol{a},\boldsymbol{v})^{\bot_e}$} & \tabincell{c}{$1\leq k\leq \lfloor \frac{p^{e'}+n-1-deg(h(x))}{p^{e'}+1} \rfloor$, \\$deg(h(x))\leq n-m-1$} & Theorem \ref{th.ConA} 1)\\
    
        2 &   $q=p^{h}\geq 5$  & \tabincell{c}{$\gcd(e',h)=e$, $\frac{h}{e}$ is even, $m\geq 2$, \\ $\GRS_m(\boldsymbol{a},\boldsymbol{v},\infty)\subseteq \GRS_m(\boldsymbol{a},\boldsymbol{v},\infty)^{\bot_e}$} & \tabincell{c}{$1\leq k\leq \lfloor \frac{p^{e'}+n-1-deg(h(x))}{p^{e'}+1} \rfloor$,\\ $deg(h(x))\leq n-m-1$} & Theorem \ref{th.ConA} 2)\\
    
        3 &   $q=p^{h}$ is odd, $h$ is even & \tabincell{c}{$\frac{h}{\gcd(e',h)}$ is odd, \\$\GRS_m(\boldsymbol{a},\boldsymbol{v})\subseteq \GRS_m(\boldsymbol{a},\boldsymbol{v})^{\bot_{\frac{h}{2}}}$} & \tabincell{c}{$1\leq k\leq \lfloor \frac{p^{e'}+n-1-deg(h(x))}{p^{e'}+1} \rfloor$, \\$deg(h(x))\leq n-m-1$} & Theorem \ref{th.ConB.1} 1)\\
    
        4 &   $q=p^{h}$  is odd, $h$ is even & \tabincell{c}{$\frac{h}{\gcd(e',h)}$ is odd, $m\geq 2$, \\ $\GRS_m(\boldsymbol{a},\boldsymbol{v},\infty)\subseteq \GRS_m(\boldsymbol{a},\boldsymbol{v},\infty)^{\bot_{\frac{h}{2}}}$} & \tabincell{c}{$1\leq k\leq \lfloor \frac{p^{e'}+n-1-deg(h(x))}{p^{e'}+1} \rfloor$, \\$deg(h(x))\leq n-m-1$} & Theorem \ref{th.ConB.1} 2)\\
      
        5 & \tabincell{c}{$q=p^h$ is odd,\\ $2^t\mid \frac{h}{m}$, $2^t=p^e+1$} & \tabincell{c}{$n=wp^{mz}$, $1\leq w\leq p^m$, $1\leq z\leq \frac{h}{m}-1$} & $1\leq k\leq \lfloor \frac{p^e+n-1}{p^e+1} \rfloor$ & Theorem \ref{th.ConC} 1)\\
        
        6 & \tabincell{c}{$q=p^h$ is odd,\\ $2^t\mid \frac{h}{m}$, $2^t=p^e+1$} & \tabincell{c}{$n=wp^{mz}+1$, $1\leq w\leq p^m$, $1\leq z\leq \frac{h}{m}-1$} & $1\leq k\leq \lfloor \frac{p^e+n-1}{p^e+1} \rfloor$ & Theorem \ref{th.ConC} 2)\\    

  \hline
	\end{tabular}}
\end{center}
\end{table}

\section{Preliminary}\label{sec2}

Let $\C$ be an $[n,k,d]_q$ linear code. Then $\C$ is called an MDS code if $d=n-k+1$. Let's now introduce an important class of MDS codes.

Suppose that $\{a_1,a_2,\cdots,a_n\}$ are $n$ distinct elements of $\F_q$ and $\boldsymbol{a}=(a_1,a_2,\dots,a_n)$. Let $\mathbb{F}_q^*=\mathbb{F}_q\backslash \{0\}$. 
For a vector $\boldsymbol{v}=(v_1,v_2,\dots,v_n)\in (\mathbb{F}_q^*)^n$ and an integer $k\geq 0$, we define a generalized Reed-Solomon (GRS) code as
\begin{align*}
  \GRS_k(\boldsymbol{a},\boldsymbol{v})=\{(v_1f(a_1),v_2f(a_2),\dots,v_nf(a_n)):\ f(x)\in \mathbb{F}_q[x],\ \deg(f(x))\leq k-1\}.
\end{align*}
It is well known that $\GRS_k(\boldsymbol{a},\boldsymbol{v})$ is an $[n, k, n-k+1]_q$ MDS code. Usually, we call the elements $a_1,a_2,\dots,a_n$ the \emph{code locators} of $\GRS_k(\boldsymbol{a},\boldsymbol{v})$ 
and the elements $v_1,v_2,\dots,v_n$ the \emph{column multipliers} of $\GRS_k(\boldsymbol{a},\boldsymbol{v})$. Moreover, an extended GRS code, denoted by $\GRS_k(\boldsymbol{a},\boldsymbol{v}, \infty)$, is defined by
\begin{align*}
 \GRS_k(\boldsymbol{a},\boldsymbol{v}, \infty)=\{(v_1f(a_1),v_2f(a_2),\dots,v_nf(a_n), f_{k-1}):\ f(x)\in \mathbb{F}_q[x], \ \deg(f(x))\leq k-1\},
\end{align*}
where $f_{k-1}$ is the coefficient of $x^{k-1}$ in $f(x)$. It is easy to show that $\GRS_k(\boldsymbol{a},\boldsymbol{v},\infty)$ is an $[n+1, k, n-k+2]_q$ MDS code.  

Recall the definitions of $e$-Galois self-orthogonal GRS and extended GRS codes. Let $\GRS_k(\boldsymbol{a},\boldsymbol{v})^{\bot_e}$ 
(resp. $\GRS_k(\boldsymbol{a},\boldsymbol{v}, \infty)^{\bot_e}$) be the $e$-Galois dual code of $\GRS_k(\boldsymbol{a},\boldsymbol{v})$ 
(resp. $\GRS_k(\boldsymbol{a},\boldsymbol{v}, \infty)$). Then $\GRS_k(\boldsymbol{a},\boldsymbol{v})$ (resp. $\GRS_k(\boldsymbol{a},\boldsymbol{v}, \infty)$) 
is $e$-Galois self-orthogonal if $\GRS_k(\boldsymbol{a},\boldsymbol{v})\subseteq \GRS_k(\boldsymbol{a},\boldsymbol{v})^{\bot_e}$ 
(resp. $\GRS_k(\boldsymbol{a},\boldsymbol{v},\infty)\subseteq \GRS_k(\boldsymbol{a},\boldsymbol{v},\infty)^{\bot_e}$). Equivalently, we can also say 
$\GRS_k(\boldsymbol{a},\boldsymbol{v})$ (resp. $\GRS_k(\boldsymbol{a},\boldsymbol{v}, \infty)$) is $e$-Galois self-orthogonal if 
$\Hull_e(\GRS_k(\boldsymbol{a},\boldsymbol{v}))=\GRS_k(\boldsymbol{a},\boldsymbol{v})$ 
(resp. $\Hull_e(\GRS_k(\boldsymbol{a},\boldsymbol{v},\infty))=\GRS_k(\boldsymbol{a},\\\boldsymbol{v},\infty))$.


We now consider the $e$-Galois hulls of the GRS cods and extended GRS codes further. To this end, for $1\leq i\leq n$, let
\begin{align}\label{equation_ui}
    u_i=\prod_{1\leq j\leq n,i\neq j}(a_i-a_j)^{-1},
\end{align}
which will be need in the sequel. In addition, we need the following important results.


\begin{lemma}\label{lem2.1}(\cite[]{RefJ29}, Propositions \Rmnum{2}.1 and \Rmnum{2}.2)
Let notations be the same as before.
\begin{enumerate}
\item[\rm 1)] Let $\boldsymbol{c}=(v_1f(a_1),v_2f(a_2),\dots,v_nf(a_n))\in \GRS_k(\boldsymbol{a},\boldsymbol{v})$, then $\boldsymbol{c}\in \GRS_k(\boldsymbol{a},\boldsymbol{v})^{\bot_e}$ if and only if there exists a polynomial $g(x)\in \mathbb{F}_q[x]$ with $\deg(g(x))\leq n-k-1$ such that
\begin{align*}
  (v_1^{p^e+1}f^{p^e}(a_1),v_2^{p^e+1}f^{p^e}(a_2),\dots,v_n^{p^e+1}f^{p^e}(a_n))=(u_1g(a_1),u_2g(a_2),\dots,u_ng(a_n)).
\end{align*} 
\item [\rm 2)]	Let $\boldsymbol{c}=(v_1f(a_1),v_2f(a_2),\dots,v_nf(a_n),f_{k-1})\in \GRS_k(\boldsymbol{a},\boldsymbol{v},\infty)$, then $\boldsymbol{c}\in \GRS_k(\boldsymbol{a},\boldsymbol{v},\infty)^{\bot_e}$ if and only if there exists a polynomial $g(x)\in \mathbb{F}_q[x]$ with $\deg(g(x))\leq n-k$ such that
\begin{align*}
  (v_1^{p^e+1}f^{p^e}(a_1),v_2^{p^e+1}f^{p^e}(a_2),\dots,v_n^{p^e+1}f^{p^e}(a_n),f_{k-1}^{p^e})=(u_1g(a_1),u_2g(a_2),\dots,u_ng(a_n),-g_{n-k}),
\end{align*}
where $g_{n-k}$ is the coefficient of $x^{n-k}$ in $g(x)$.
\end{enumerate}
\end{lemma}

\begin{lemma}\label{lem.gcd}(\cite{RefJ37}, Lemma 3)
  Let $s\geq 1$ and $p>1$ be two integers. Then
  \begin{equation}\label{equ.gcd}
  \gcd(p^r+1, p^s-1)=\left\{
  \begin{array}{rcl}
  1 & & {if\ \frac{s}{\gcd(r,s)}\ is\ odd\ and\ p\ is\ even,}\\
  2 & & {if\ \frac{s}{\gcd(r,s)}\ is\ odd\ and\ p\ is\ odd,}\\
  p^{\gcd(r,s)}+1 & & {if\ \frac{s}{\gcd(r,s)}\ is\ even.}\\
  \end{array} \right.
  \end{equation}
\end{lemma}

\section{Constructions} \label{sec3}

In this section, we construct several classes of MDS codes with Galois hulls of arbitrary dimensions via (extended) GRS codes. We also give two general methods to construct MDS codes with Galois hulls 
of arbitrary dimensions from Hermitian or general Galois self-orthogonal (extended) GRS codes. Some MDS codes with Galois hulls of arbitrary dimensions constructed from Hermitian self-orthogonal (extended) 
GRS codes have larger dimensions.

\subsection{MDS codes with Galois hulls of arbitrary dimensions from Galois self-orthogonal (extended) GRS codes}\label{Construction A}


In this subsection, we use Galois self-orthogonal (extended) GRS codes to construct MDS codes with Galois hulls of arbitrary dimensions. We begin with the following lemmas, which give some necessary conditions 
for Galois self-orthogonal (extended) GRS codes.

\begin{lemma}\label{coro.Galois self-orthogonal GRS}
  If $\GRS_k(\boldsymbol{a},\boldsymbol{v})\subseteq \GRS_k(\boldsymbol{a},\boldsymbol{v})^{\bot_e}$, then there exists a monic polynomial $h(x)\in \mathbb{F}_q[x]$ with $\deg(h(x))\leq n-k-1$ such that 
  \begin{equation}\label{eq.coro.Galois self-orthogonal GRS}
      \lambda u_ih(a_i)=v_i^{p^e+1},\ 1\leq i\leq n,
  \end{equation}
  where $\lambda \in \mathbb{F}_q^*$.
\end{lemma}
\begin{proof}
Suppose that $\GRS_k(\boldsymbol{a},\boldsymbol{v})\subseteq \GRS_k(\boldsymbol{a},\boldsymbol{v})^{\bot_e}$. Then, for any codeword 
$$\boldsymbol{c}=(v_1f(a_1),v_2f(a_2),\dots,v_nf(a_n)) \in  \GRS_k(\boldsymbol{a},\boldsymbol{v}),$$
we have $\boldsymbol{c}\in \GRS_k(\boldsymbol{a},\boldsymbol{v})^{\bot_e}$. Specially, taking $f(x)=1$, then $\boldsymbol{c}=(v_1,v_2,\dots,v_n)$. According to the result 1) of Lemma \ref{lem2.1}, there exists a polynomial $g(x)\in \mathbb{F}_q[x]$ with $\deg(g(x))\leq n-k-1$ such that $v_i^{p^e+1}f^{p^e}(a_i)=u_ig(a_i)$, i.e., $v_i^{p^e+1}=u_ig(a_i)$ for $1\leq i\leq n$. It is clear that there exists a $\lambda^{-1}\in \mathbb{F}_q^*$ such that $h(x)=\lambda^{-1}g(x)\in \mathbb{F}_q[x]$ is a monic polynomial with $\deg(h(x))=\deg(g(x))\leq n-k-1$. Note that $u_i g(a_i)=\lambda u_i h(a_i)$, hence $\lambda u_ih(a_i)=v_i^{p^e+1}$ for $1\leq i\leq n$. This completes the proof. 
\end{proof}

\begin{lemma}\label{coro.Galois self-orthogonal EGRS}
  If $\GRS_k(\boldsymbol{a},\boldsymbol{v},\infty)\subseteq 
  \GRS_k(\boldsymbol{a},\boldsymbol{v},\infty)^{\bot_e}$ with $k\geq 2$, then there exists a monic polynomial $h(x)\in \mathbb{F}_q[x]$ with $\deg(h(x))\leq n-k-1$ such that 
  \begin{equation}\label{eq.coro.Galois self-orthogonal EGRS}
      \lambda u_ih(a_i)=v_i^{p^e+1},\ 1\leq i\leq n,
  \end{equation}
  where $\lambda \in \mathbb{F}_q^*$.
\end{lemma}
\begin{proof}
  Similar to the proof of Lemma \ref{coro.Galois self-orthogonal GRS}, take $f(x)=1$, then 
  $$\boldsymbol{c}=(v_1,v_2,\dots,v_n,0)\in 
  \GRS_k(\boldsymbol{a},\boldsymbol{v},\infty) \subseteq \GRS_k(\boldsymbol{a},\boldsymbol{v},\infty)^{\bot_e}.$$ 
  According to the result 2) of Lemma \ref{lem2.1}, there exists a polynomial $g(x)\in \mathbb{F}_q[x]$ with $\deg(g(x))\leq n-k$  such that $v_i^{p^e+1}=u_ig(a_i)$ for $1\leq i\leq n$ and $g_{n-k}=-f_{k-1}^{p^e}=0$. It follows that 
  $\deg(g(x))\leq n-k-1$. Clearly, there exists a $\lambda^{-1}\in \mathbb{F}_q^*$ such that $h(x)=\lambda^{-1}g(x)\in \mathbb{F}_q[x]$ is a monic polynomial with $\deg(h(x))=\deg(g(x))\leq n-k-1$. Therefore, $v_i^{p^e+1}=u_ig(a_i)=\lambda u_ih(a_i)$ for $1\leq i\leq n$. This completes the proof. 
\end{proof}

\begin{remark}\label{rem.1}
    
    \begin{enumerate}
    \item [\rm 1)]  By the proofs of Lemmas \ref{coro.Galois self-orthogonal GRS} and \ref{coro.Galois self-orthogonal EGRS}, $h(x)=\lambda^{-1}g(x)$ and $g(a_i)=u_i^{-1}v_i^{p^e+1}$, $1\leq i\leq n$ 
    when $f(x)=1$. Note that $g(x)$ is a polynomial with $\deg(g(x))\leq n-k-1$ and $a_i$, $u_i$ and $v_i$ are all known for $1\leq i\leq n$ when a GRS code or extended GRS code is given. Hence, 
    according to the existence and uniqueness of Lagrange Interpolation Formula, the interpolation polynomial obtained by using these $n$ mutually different interpolation points 
    $(a_1,u_1^{-1}v_1^{p^e+1})$, $(a_2,u_2^{-1}v_2^{p^e+1})$, $\cdots$, $(a_n,u_n^{-1}v_n^{p^e+1})$ is actually $g(x)$ itself. Therefore the structure of $h(x)$ and $\deg(h(x))$ can be 
    easily derived from $h(x)=\lambda^{-1}g(x)$.

    \item [\rm 2)]  In practical applications, there are many occasions where we don't even need to use Lagrange Interpolation Formula. Because in the constructions of many known self-orthogonal 
    (extended) GRS codes, $g(x)$, as an important intermediate tool, was usually given directly  (e.g., see \cite[]{RefJ29,RefJ30,RefJ24,RefJ7,RefJ33,RefJ19,RefJ6,RefJ110,RefJ111} and references therein). 
    Hence taking $f(x)=1$, in these cases, the structure of $h(x)$ and $\deg(h(x))$ can be determined directly. Some specific examples will be given later.
    \end{enumerate}
    \end{remark}
  
We now can start our constructions. It is worth noting that the polynomial $h(x)$ appearing in the constructions refers to $h(x)$ in Lemmas \ref{coro.Galois self-orthogonal GRS} and 
\ref{coro.Galois self-orthogonal EGRS}, i.e., Eqs. (\ref{eq.coro.Galois self-orthogonal GRS}) and (\ref{eq.coro.Galois self-orthogonal EGRS}) hold. According to Remark \ref{rem.1}, 
$\deg(h(x))$ can be easily determined. 

\begin{theorem}\label{th.ConA}
Let $q=p^h\geq5$ be a prime power. Let $1\leq e, e'\leq h-1$ such that $e=\gcd(e',h)$ and $\frac{h}{e}$ is even. Then the following hold.
\begin{enumerate}
\item [\rm 1)] Suppose that $\GRS_m(\boldsymbol{a},\boldsymbol{v})\subseteq \GRS_m(\boldsymbol{a},\boldsymbol{v})^{\bot_e}$, then for $1\leq k\leq \lfloor \frac{p^{e'}+n-1-\deg(h(x))}{p^{e'}+1} \rfloor$, there exists an $[n, k]_{q}$ MDS code $\C$ with $\dim(\Hull_{e'}(\C))=l$, where $0\leq l\leq k$.
\item [\rm 2)] Suppose that $\GRS_m(\boldsymbol{a},\boldsymbol{v},\infty)\subseteq \GRS_m(\boldsymbol{a},\boldsymbol{v},\infty)^{\bot_e}$ with $m\geq 2$, then for $1\leq k\leq \lfloor \frac{p^{e'}+n-1-\deg(h(x))}{p^{e'}+1} \rfloor$, there exists an $[n+1,k]_{q}$ MDS code $\C$ with $\dim(\Hull_{e'}(\C))=l$, where $0\leq l\leq k-1$.
\end{enumerate}
\end{theorem}
\begin{proof}
  1) It follows from $\GRS_m(\boldsymbol{a},\boldsymbol{v})\subseteq \GRS_m(\boldsymbol{a},\boldsymbol{v})^{\bot_e}$ and Lemma \ref{coro.Galois self-orthogonal GRS} that there exists a monic polynomial 
  $h(x)\in \mathbb{F}_q[x]$ with $\deg(h(x))\leq n-m-1$ such that 
  \begin{align}\label{eq.ConA.1}
    \lambda u_ih(a_i)=v_i^{p^e+1}\neq 0,\ 1\leq i\leq n,
  \end{align}
  where $\lambda \in \mathbb{F}_q^*$. Since $e=\gcd(e',h)$ and $\frac{h}{e}$ is even, by Lemma \ref{lem.gcd}, 
  \begin{align*}
    \gcd(p^{e'}+1,p^h-1)=p^{\gcd(e',h)}+1=p^e+1.
  \end{align*}
  Therefore, there exist two integers $\mu$ and $\nu$ such that $\mu (p^{e'}+1)+\nu (p^h-1)=p^e+1$. Denote $v_i^{\mu}=v_i'$, then
  \begin{align}\label{eq.ConA.2}
    v_i^{p^e+1}=v_i^{\mu (p^{e'}+1)+\nu (p^h-1)}=(v_i^{\mu})^{p^{e'}+1}=(v_i')^{p^{e'}+1}.
  \end{align}
  Substituting Eq. (\ref{eq.ConA.2}) into Eq. (\ref{eq.ConA.1}), we have 
\begin{equation}\notag
 \lambda u_ih(a_i)=(v_i')^{p^{e'}+1},\ 1\leq i\leq n.
\end{equation}

Since $(q-1)\nmid (p^{e'}+1)$, there is an $\alpha \in \mathbb{F}_{q}^*$ such that $\beta=\alpha^{p^{e'}+1}\neq 1$. Let $\boldsymbol{a}=(a_1,a_2,\dots,a_n)$ be the same as before and $\boldsymbol{v'}=(\alpha v_1',\dots,\alpha v_s',v_{s+1}',\dots,v_n')\in (\mathbb{F}_q^*)^n$, where $s=k-l\leq k$. We now consider the $e'$-Galois hull of the $[n, k]_{q}$ MDS code $\C=\GRS_k(\boldsymbol{a},\boldsymbol{v'})$. 

For any codeword 
$$\boldsymbol{c}=(\alpha v_1'f(a_1),\dots,\alpha v_s'f(a_s), v_{s+1}'f(a_{s+1}),\dots,v_n'f(a_n))\in \Hull_{e'}(\C),$$
 where $f(x)\in \mathbb{F}_{q}[x]$ and $\deg(f(x))\leq k-1$. By the result 1) of Lemma \ref{lem2.1}, there exists a polynomial $g(x)\in \mathbb{F}_{q}[x]$ with $\deg(g(x))\leq n-k-1$ such that
\begin{align*}
  & (\alpha^{p^{e'}+1} (v_1')^{p^{e'}+1}f^{p^{e'}}(a_1),\dots,\alpha^{p^{e'}+1} (v_s')^{p^{e'}+1}f^{p^{e'}}(a_s),(v_{s+1}')^{p^{e'}+1}f^{p^{e'}}(a_{s+1}),
  \dots, (v_n')^{p^{e'}+1}f^{p^{e'}}(a_n))\\
= & (u_1g(a_1),\dots,u_sg(a_s),u_{s+1}g(a_{s+1}),\dots,u_ng(a_n)).
\end{align*}
Replacing $\alpha^{p^{e'}+1}$ and $(v_i')^{p^{e'}+1}$ with $\beta$ and $\lambda u_ih(a_i)$, respectively, we have
\begin{equation}\label{eq.ConA.3}
\begin{split}
& (\beta \lambda u_1h(a_1)f^{p^{e'}}(a_1),\dots,\beta \lambda u_sh(a_s)f^{p^{e'}}(a_s),\lambda u_{s+1}h(a_{s+1})f^{p^{e'}}(a_{s+1}),\dots, \lambda u_nh(a_n)f^{p^{e'}}(a_n))\\
= & (u_1g(a_1),\dots,u_sg(a_s),u_{s+1}g(a_{s+1}),\dots,u_ng(a_n)).
\end{split}
\end{equation}

On one hand, from the last $n-s$ coordinates of Eq. (\ref{eq.ConA.3}), we have 
\begin{align*}
    \lambda u_ih(a_i)f^{p^{e'}}(a_i)=u_ig(a_i),\ s+1\leq i\leq n,
\end{align*}
i.e., $\lambda h(x)f^{p^{e'}}(x)=g(x)$ has at least $n-s$ distinct roots. Since $\deg(h(x))\leq n-m-1$ and $1\leq k\leq \lfloor \frac{p^{e'}+n-1-\deg(h(x))}{p^{e'}+1} \rfloor$,
\begin{align*}
& \deg(h(x)f^{p^{e'}}(x))\leq \deg(h(x))+p^{e'}(k-1)\leq n-k-1,\\
& \deg(g(x))\leq n-k-1.
\end{align*}
Recall that $s=k-l\leq k$, then
\begin{equation}\label{eq_xiugai_1}
 \deg(\lambda h(x)f^{p^{e'}}(x)-g(x))\leq n-k-1\leq n-s-1< n-s.
\end{equation}
That is, $\lambda h(x)f^{p^{e'}}(x)=g(x)$.

On the other hand, from the first $s$ coordinates of Eq. (\ref{eq.ConA.3}), we have
\begin{equation}\notag
  \beta \lambda u_ih(a_i)f^{p^{e'}}(a_i)=u_ig(a_i)=\lambda u_i h(a_i)f^{p^{e'}}(a_i),\ 1\leq i\leq s.
\end{equation}
Therefore, $f(a_i)=0$ ($1\leq i\leq s$) for $\beta\neq 1$ and $\lambda u_ih(a_i)\neq 0$. It follows that $f(x)$ can be written as
\begin{equation}\notag
  f(x)=r(x)\prod_{i=1}^{s}(x-a_i),
\end{equation}
for some $r(x)\in \mathbb{F}_{q}[x]$ with $\deg(r(x))\leq k-1-s=l-1$. It deduces that $\dim(\Hull_{e'}(\C))\leq l$.

Conversely, let $f(x)=r(x)\prod_{i=1}^{s}(x-a_i)$, where $r(x)\in \mathbb{F}_{q}[x]$ with $\deg(r(x))\leq k-1-s=l-1$. Take $g(x)=\lambda h(x)f^{p^{e'}}(x)$, 
then $\deg(g(x))\leq \deg(h(x))+p^{e'}(k-1)\leq n-k-1$ and 
\begin{align*}
    & (\alpha^{p^{e'}+1} (v_1')^{p^{e'}+1}f^{p^{e'}}(a_1),\dots,\alpha^{p^{e'}+1} (v_s')^{p^{e'}+1}f^{p^{e'}}(a_s),(v_{s+1}')^{p^{e'}+1}f^{p^{e'}}(a_{s+1}),
    \dots, (v_n')^{p^{e'}+1}f^{p^{e'}}(a_n))\\
  = & (u_1g(a_1),\dots,u_sg(a_s),u_{s+1}g(a_{s+1}),\dots,u_ng(a_n)).
  \end{align*}
 By the result 1) of Lemma \ref{lem2.1}, the vector
 $$(\alpha v_1'f(a_1),\dots,\alpha v_s'f(a_s), v_{s+1}'f(a_{s+1}),\dots,v_n'f(a_n))\in \Hull_{e'}(\C).$$ 
 It deduces that $\dim(\Hull_{e'}(\C))\geq l$.

Combining both aspects, we have $\dim(\Hull_{e'}(\C))=l$, where $0\leq l\leq k$. The desired result follows.

2) Since $\GRS_m(\boldsymbol{a},\boldsymbol{v},\infty)\subseteq \GRS_m(\boldsymbol{a},\boldsymbol{v},\infty)^{\bot_e}$ with $m\geq 2$ and by Lemma \ref{coro.Galois self-orthogonal EGRS}, 
there exists a monic polynomial $h(x)\in \mathbb{F}_q[x]$ with $\deg(h(x))\leq n-m-1$ such that $\lambda u_ih(a_i)=v_i^{p^e+1}$, $1\leq i\leq n$. Similar to the result 1), it can be proved that there exists $v_i'\in \mathbb{F}_{q}^*$ such that
\begin{equation}\notag
  (v_i')^{p^{e'}+1}=\lambda u_ih(a_i), 1\leq i\leq n.
\end{equation}

Choose $\alpha \in \mathbb{F}_{q}^*$ satisfying $\beta=\alpha^{p^{e'}+1}\neq 1$ again. Let $\boldsymbol{a}=(a_1,a_2,\dots,a_n)$ be the same as before and $\boldsymbol{v'}=(\alpha v_1',\dots,\alpha v_s',v_{s+1}',\dots,v_n')$, where $s=k-l-1\leq k-1$. We now consider the $e'$-Galois hull of the $[n+1,k]_{q}$ MDS code $\C=\GRS_k(\boldsymbol{a},\boldsymbol{v'},\infty)$. 

For any codeword
 $$\boldsymbol{c}=(\alpha v_1'f(a_1),\dots,\alpha v_s'f(a_s), v_{s+1}'f(a_{s+1}),\dots,v_n'f(a_n),f_{k-1})\in \Hull_{e'}(\C),$$
where $f(x)\in \mathbb{F}_{q}[x]$ and $\deg(f(x))\leq k-1$. By the result 2) of Lemma \ref{lem2.1}, and replacing $\alpha^{p^{e'}+1}$ and $(v_i')^{p^{e'}+1}$ with $\beta$ and $\lambda u_ih(a_i)$, respectively, there exists a polynomial $g(x)\in \mathbb{F}_{q}[x]$ with $\deg(g(x))\leq n-k$ such that
\begin{equation}\label{eq.ConA.4}
\begin{split}
& (\beta\lambda u_1h(a_1)f^{p^{e'}}(a_1),\dots,\beta\lambda u_sh(a_s)f^{p^{e'}}(a_s), \lambda u_{s+1}h(a_{s+1})f^{p^{e'}}(a_{s+1}),\dots,\lambda u_nh(a_n)f^{p^{e'}}(a_n), \\
& f_{k-1}^{p^{e'}}) = (u_1g(a_1),\dots,u_sg(a_s),u_{s+1}g(a_{s+1}),\dots,u_ng(a_n),-g_{n-k}).
\end{split}
\end{equation}

On one hand, from the last $n-s+1$ coordinates of Eq. (\ref{eq.ConA.4}), $\lambda u_ih(a_i)f^{p^{e'}}(a_i)=u_ig(a_i)$, for $s+1\leq i\leq n$
and $f_{k-1}^{p^{e'}}=-g_{n-k}$. It follows from $\deg(h(x))\leq n-m-1$, $1\leq k\leq \lfloor \frac{p^{e'}+n-1-\deg(h(x))}{p^{e'}+1} \rfloor$ and $s=k-l-1\leq k-1$ that $\lambda h(x)f^{p^{e'}}(x)=g(x)$. 
We now determine the value of $f_{k-1}$. If $f_{k-1}\neq 0$, then $\deg(h(x))+p^{e'}(k-1)=n-k$, which contradicts to $1\leq k\leq \lfloor \frac{p^{e'}+n-1-\deg(h(x))}{p^{e'}+1} \rfloor$. Therefore, 
$f_{k-1}=0$ and $\deg(f(x))\leq k-2$.

On the other hand, from the first $s$ coordinates of Eq. (\ref{eq.ConA.4}), we have
\begin{equation}\notag
  \beta\lambda u_ih(a_i)f^{p^{e'}}(a_i)=u_ig(a_i)=\lambda u_ih(a_i)f^{p^{e'}}(a_i),\ 1\leq i\leq s.
\end{equation}
Therefore, $f(a_i)=0$ ($1\leq i\leq s$) for $\beta\neq 0$ and $\lambda u_ih(a_i)\neq 0$. It follows that $f(x)$ can be written as
\begin{equation}\notag
  f(x)=r(x)\prod_{i=1}^{s}(x-a_i),
\end{equation}
for some $r(x)\in \mathbb{F}_{q}[x]$ with $\deg(r(x))\leq k-2-s=l-1$. It deduces that $\dim(\Hull_{e'}(\C))\leq l$.

Conversely, let $f(x)=r(x)\prod_{i=1}^{s}(x-a_i)$, where $r(x)\in \mathbb{F}_{q}[x]$ with $\deg(r(x))\leq k-2-s=l-1$. Take $g(x)=\lambda h(x)f^{p^{e'}}(x)$, 
then $\deg(g(x))\leq \deg(h(x))+p^{e'}(k-2)\leq n-k-1$ and 
\begin{align*}
    & (\alpha^{p^{e'}+1} (v_1')^{p^{e'}+1}f^{p^{e'}}(a_1),\dots,\alpha^{p^{e'}+1} (v_s')^{p^{e'}+1}f^{p^{e'}}(a_s),(v_{s+1}')^{p^{e'}+1}f^{p^{e'}}(a_{s+1}),
    \dots, (v_n')^{p^{e'}+1}f^{p^{e'}}(a_n),0)\\
  = & (u_1g(a_1),\dots,u_sg(a_s),u_{s+1}g(a_{s+1}),\dots,u_ng(a_n),0).
  \end{align*}
 By the result 2) of Lemma \ref{lem2.1}, the vector 
 $$(\alpha v_1'f(a_1),\dots,\alpha v_s'f(a_s), v_{s+1}'f(a_{s+1}),\dots,v_n'f(a_n),0)\in \Hull_{e'}(\C).$$ It deduces that $\dim(\Hull_{e'}(\C))\geq l$

Combining both aspects, we have $\dim(\Hull_{e'}(\C))=l$, where $0\leq l\leq k-1$. The desired result follows. 
\end{proof}

\begin{remark}\label{rem2}
\begin{enumerate}
\item [\rm 1)] In the light of our present knowledge, there is relatively little work on general Galois self-orthogonal (extended) GRS codes. Classes $10,\ 12,\ 15$ and $19$ in Table \ref{tab:1} are some examples.	
\item [\rm 2)] For the case $e=0$, a similar result was given by Theorems 4.1 and 4.2 of \cite{RefJ30}. We have list them as Classes $19$ and $20$ in Table \ref{tab:1}. 
\end{enumerate}
\end{remark}

\begin{table}[!htb]
    \centering
    \caption{Some examples satisfying the conditions that $\gcd(e',h)=e$ and $\frac{h}{e}$ is even}
    \label{tab:3}       
      \begin{tabular}{c|c|l|c|c|l}
        \hline
        $h$ & $e$ & $e'$ & $h$ & $e$ & $e'$ \\
        \hline
          2 & 1 & 1 & 4 & 1 & 1, 3\\
          4 & 2 & 2 & 6 & 1 & 1, 5 \\
          6 & 3 & 3 & 8 & 1 & 1, 3, 5, 7\\
          8 & 2 & 2, 6 & 8 & 4 & 4\\
          10 & 1 & 1, 3, 7, 9 & 10 & 5 & 5\\
          12 & 1 & 1, 5, 7, 11 & 12 & 2 & 2, 10\\
          12 & 3 & 3, 9 & 12 & 6 & 6\\
        \hline
      \end{tabular}
    \end{table}

\begin{example}\label{example.thA.1}
    In Table \ref{tab:3}, we give some examples satisfying the conditions that $e=\gcd(e',h)$ and $\frac{h}{e}$ is even. By Theorem \ref{th.ConA}, we can obtain more 
    MDS codes with $e'$-Galois hulls of arbitrary dimensions from Classes $10,\ 12,\ 15$ and $19$. Taking Class $15$ as an example, we show the specific steps as follows:
    
    $\bullet \textbf{Step 1. Obtain e-Galois self-orthogonal codes.}$ 
    
    For example, taking $(p,h,t,e)=(3,8,2,1)$ in Class $15$ (i.e., Theorem \Rmnum{3}.2 of \cite[]{RefJ29}), we know that $[6561,k]_{3^8}$ $l$-dim $1$-Galois hull GRS codes exist, where 
    $1\leq k\leq 1640$ and $0\leq l\leq k$. Take $l=k=1640$, then $[6561,1640]_{3^8}$ is a $1$-Galois self-orthogonal GRS code. 
    
    $\bullet \textbf{Step 2. Determine {\bf{deg}}(h(x)).}$ 

    According to Remark \ref{rem.1}, $h(x)=\lambda^{-1}g(x)$. Hence, we only need to determine the structure of $g(x)$. This is easy to do because, as the result 2) of Remark \ref{rem.1} 
    said, $g(x)=f(x)^{p^e}$ was given explicitly in the proof of Theorem \Rmnum{3}.2 of \cite[]{RefJ29}. Taking $f(x)=1$ further, we have $g(x)=f(x)^{p^e}=1$. Therefore, we can directly 
    determine that $h(x)=\lambda^{-1}g(x)=1$ and $\deg(h(x))=0$, where $\lambda=1\in \mathbb{F}_{3^8}^*$. 
    
    $\bullet \textbf{Step 3. Derive new MDS codes with e'-Galois hulls of arbitrary dimensions.}$ 
    
    By the result 1) of Theorem \ref{th.ConA} and Table \ref{tab:3}, MDS codes with $1,\ 3,\ 5$ and $7$-Galois hulls of arbitrary dimensions can be derived as follows:

    \begin{itemize}
      \item	$[6561, k_1]_{3^8}$ MDS code with $l$-dim $1$-Galois hull for $1\leq k_1\leq 1640$, where $0\leq l\leq k_1$;
      \item $[6561, k_3]_{3^8}$ MDS code with $l$-dim $3$-Galois hull for $1\leq k_3\leq 235$, where $0\leq l\leq k_3$;
      \item $[6561, k_5]_{3^8}$ MDS code with $l$-dim $5$-Galois hull for $1\leq k_5\leq 27$, where $0\leq l\leq k_5$;
      \item $[6561, k_7]_{3^8}$ MDS code with $l$-dim $7$-Galois hull for $1\leq k_7\leq 3$, where $0\leq l\leq k_7$;
      \end{itemize}
    
    Note that $2\times 3=6$, $2\times 5=10$, $2\times 7=14$ are not divisors of $h=8$, hence they cannot be constructed by 
    Classes $10,\ 12,\ 15$ and $19$, which implies that all of these MDS codes with Galois hulls of arbitrary dimensions are new. 
\end{example}

\subsection{MDS codes with Galois hulls of arbitrary dimensions from Hermitian self-orthogonal (extended) GRS codes}\label{constructionB}

  Throughout this subsection, we take $h$ as an even positive integer and $e=\frac{h}{2}$. For any prime $p$, let $q=p^h$, then $\sqrt{q}$ is still a prime power. 
  Since the Galois hull is a generalization of the Hermitian hull, we can conclude that $\GRS_k(\boldsymbol{a},\boldsymbol{v})$ (resp. $\GRS_k(\boldsymbol{a},\boldsymbol{v},\infty)$) 
  is Hermitian self-orthogonal if $\GRS_k(\boldsymbol{a},\boldsymbol{v})\subseteq \GRS_k(\boldsymbol{a},\boldsymbol{v})^{\bot_{\frac{h}{2}}}$ 
  (resp. $\GRS_k(\boldsymbol{a},\boldsymbol{v},\infty)\subseteq \GRS_k(\boldsymbol{a},\boldsymbol{v},\infty)^{\bot_{\frac{h}{2}}}$). In this subsection, 
  our main goal is to construct MDS codes with Galois hulls of arbitrary dimensions from Hermitian self-orthogonal (extended) GRS codes. To this end, we need the following lemma.
 
  \begin{lemma}\label{coro.Hermitian self-orthogonal GRS and EGRS}
   Let $q=p^h$, where $p$ is a prime and $h$ is an even positive integer.
    \begin{enumerate}
    \item [\rm 1)]	If $\GRS_k(\boldsymbol{a},\boldsymbol{v})\subseteq \GRS_k(\boldsymbol{a},\boldsymbol{v})^{\bot_{\frac{h}{2}}}$, then there exists a monic polynomial $h(x)\in \mathbb{F}_q[x]$ with 
    $\deg(h(x))\\\leq n-k-1$ such that $\lambda u_ih(a_i)=v_i^{\sqrt{q}+1}$, $1\leq i\leq n$, where $\lambda\in \mathbb{F}_{q}^*$.
    \item [\rm 2)] If $\GRS_k(\boldsymbol{a},\boldsymbol{v},\infty)\subseteq \GRS_k(\boldsymbol{a},\boldsymbol{v},\infty)^{\bot_{\frac{h}{2}}}$ with $k\geq 2$, then there exists a monic polynomial 
    $h(x)\in \mathbb{F}_q[x]$ with $\deg(h(x))\leq n-k-1$ such that $\lambda u_ih(a_i)=v_i^{\sqrt{q}+1}$, $1\leq i\leq n$, where $\lambda\in \mathbb{F}_{q}^*$.
    \end{enumerate}
  \end{lemma}
\begin{proof}
	The proof of this lemma is similar to Lemma \ref{coro.Galois self-orthogonal GRS}, so it is omitted.
\end{proof}

\begin{remark}\label{rem.3}
    Similar to Remark \ref{rem.1}, the structure and degree of $h(x)$ here can also be easily determined.
\end{remark}

  \begin{theorem}\label{th.ConB.1}
  Let $q=p^h$ be an odd prime power, where $h$ is an even positive integer. Let $0\leq e'\leq h-1$ and such that $\frac{h}{\gcd(e',h)}$ is odd. Then the following hold.
  \begin{enumerate}
  \item [\rm 1)]Suppose $\GRS_m(\boldsymbol{a},\boldsymbol{v})\subseteq \GRS_m(\boldsymbol{a},\boldsymbol{v})^{\bot_{\frac{h}{2}}}$, then for $1\leq k\leq \lfloor \frac{p^{e'}+n-1-\deg(h(x))}{p^{e'}+1} \rfloor$, 
  there exists an $[n, k]_{q}$ MDS code $\C$ with $\dim(\Hull_{e'}(\C))=l$, where $0\leq l\leq k$.
  \item [\rm 2)]Suppose $\GRS_m(\boldsymbol{a},\boldsymbol{v},\infty)\subseteq \GRS_m(\boldsymbol{a},\boldsymbol{v},\infty)^{\bot_{\frac{h}{2}}}$ with $m\geq 2$, then for 
  $1\leq k\leq \lfloor \frac{p^{e'}+n-1-\deg(h(x))}{p^{e'}+1} \rfloor$, there exists an $[n+1,k]_{q}$ MDS code $\C$ with $\dim(\Hull_{e'}(\C))=l$, where $0\leq l\leq k-1$.	
  \end{enumerate}
  \end{theorem}
  \begin{proof} 1) From $\GRS_m(\boldsymbol{a},\boldsymbol{v})\subseteq \GRS_m(\boldsymbol{a},\boldsymbol{v})^{\bot_{\frac{h}{2}}}$ and the result 1) of Lemma
  \ref{coro.Hermitian self-orthogonal GRS and EGRS}, there exists a monic polynomial $h(x)\in \mathbb{F}_q[x]$ with $\deg(h(x))\leq n-m-1$ such that
  \begin{align}\label{eq.ConB.1}
    \lambda u_ih(a_i)=v_i^{\sqrt{q}+1}\neq 0,\ 1\leq i\leq n,
  \end{align} 
  where $\lambda\in \mathbb{F}_{q}^*$. Note that $(v_i^{\sqrt{q}+1})^{\sqrt{q}}=v_i^{q-1+\sqrt{q}+1}=v_i^{\sqrt{q}+1}$, 
  then $v_i^{\sqrt{q}+1}\in \mathbb{F}_{\sqrt{q}}\subseteq \mathbb{F}_{q}$ for $1\leq i\leq n$. Hence there exists $v_i'\in \mathbb{F}_{q}^*$ such that
  \begin{equation}\label{eq.ConB.2}
    v_i^{\sqrt{q}+1}=v_i'^2,\ 1\leq i\leq n.
  \end{equation}
 Since both $\frac{h}{\gcd(e',h)}$ and $p$ are odd, by Lemma \ref{lem.gcd}, $\gcd(p^{e'}+1, p^{h}-1)=2$. Therefore, there exist two integers $\mu$ and $\nu$ such that 
  $\mu (p^{e'}+1)+\nu (p^{h}-1)=2$. Substituting it into Eqs. (\ref{eq.ConB.1}) and (\ref{eq.ConB.2}), we can get $(v_i')^{\mu (p^{e'}+1)}=\lambda u_ih(a_i),\ 1\leq i\leq n$. 
  Denote $v_i''=(v_i')^{\mu},\ 1\leq i\leq n$, then
  \begin{equation}\notag
    (v_i'')^{p^{e'}+1}=\lambda u_ih(a_i),\ 1\leq i\leq n.
  \end{equation}
  
  Since $(q-1)\nmid (p^{e'}+1)$, there is an $\alpha \in \mathbb{F}_{q}^*$ such that $\beta=\alpha^{p^{e'}+1}\neq 1$. 
  Let $\boldsymbol{a}=(a_1,a_2,\dots,a_n)$ be the same as before and 
  $\boldsymbol{v''}=(\alpha v_1'',\dots,\alpha v_s'',v_{s+1}'',\dots,v_n'')$, where $s=k-l\leq k$. We now consider the $e'$-Galois hull of the $[n, k]_{q}$ MDS code $\C=\GRS_k(\boldsymbol{a},\boldsymbol{v''})$. 

For any codeword 
$$\boldsymbol{c}=(\alpha v_1''f(a_1),\dots,\alpha v_s''f(a_s), v_{s+1}''f(a_{s+1}),\dots,v_n''f(a_n))\in \Hull_{e'}(\C),$$
where $f(x)\in \mathbb{F}_{q}[x]$ and $\deg(f(x))\leq k-1$. Replace $\alpha^{p^{e'}+1}$ and $(v_i'')^{p^{e'}+1}$ with $\beta$ and $\lambda u_ih(a_i)$, respectively. 
By the result 1) of Lemma \ref{lem2.1}, there exists a polynomial $g(x)\in \mathbb{F}_{q}[x]$ with $\deg(g(x))\leq n-k-1$ such that
\begin{equation}\label{eq.ConB.3}
\begin{split}
& (\beta \lambda u_1h(a_1)f^{p^{e'}}(a_1),\dots,\beta \lambda u_sh(a_s)f^{p^{e'}}(a_s),\lambda u_{s+1}h(a_{s+1})f^{p^{e'}}(a_{s+1}),\dots, \lambda u_nh(a_n)f^{p^{e'}}(a_n))\\
= & (u_1g(a_1),\dots,u_sg(a_s),u_{s+1}g(a_{s+1}),\dots,u_ng(a_n)).
\end{split}
\end{equation}
Similar to the result 1) of Theorem \ref{th.ConA}, we can deduce that $\lambda h(x)f^{p^{e'}}(x)=g(x)$ from the last $n-s$ coordinates of Eq. (\ref{eq.ConB.3}). 
Since $\beta \neq 1$ and $\lambda u_ih(a_i)\neq 0$, from the first $s$ coordinates of Eq. (\ref{eq.ConB.3}), we can derive that $f(a_i)=0$ for $1\leq i\leq s$. 
Hence, $f(x)$ can be written as
\begin{equation}\notag
  f(x)=r(x)\prod_{i=1}^{s}(x-a_i)
\end{equation}
for some $r(x)\in \mathbb{F}_{q}[x]$ with $\deg(r(x))\leq k-1-s=l-1$. It deduces that $\dim(\Hull_{e'}(\C))\leq l$.

Conversely, for any $f(x)=r(x)\prod_{i=1}^{s}(x-a_i)$, where $r(x)\in \mathbb{F}_{q}[x]$ with $\deg(r(x))\leq k-1-s=l-1$. Let $g(x)=\lambda h(x)f^{p^{e'}}(x)$. Similar to the proof 
of the result 1) of Theorem \ref{th.ConA} again, we can deduce that $\dim(\Hull_{e'}(\C))\geq l$.

Combining both aspects, we have $\dim(\Hull_{e'}(\C))=l$, where $0\leq l\leq k$. The desired result follows.

2) Since $\GRS_m(\boldsymbol{a},\boldsymbol{v},\infty)\subseteq \GRS_m(\boldsymbol{a},\boldsymbol{v},\infty)^{\bot_{\frac{h}{2}}}$ with $m\geq 2$ and by the result 2) of Lemma 
\ref{coro.Hermitian self-orthogonal GRS and EGRS}, there exists a monic polynomial $h(x)\in \mathbb{F}_q[x]$ with $\deg(h(x))\leq n-m-1$ such that 
$\lambda u_i h(a_i)=v_i^{\sqrt{q}+1}$, $1\leq i\leq n$. For the same reason with 1) above, there exists $v_i'\in \mathbb{F}_{q}^*$ such that
  \begin{equation}\notag
    (v_i')^{p^{e'}+1}=\lambda u_ih(a_i),\ 1\leq i\leq n.
  \end{equation}
  
  Let $\alpha \in \mathbb{F}_{q}^*$ such that $\beta=\alpha^{p^{e'}+1}\neq 1$. Let $\boldsymbol{a}=(a_1,a_2,\dots,a_n)$ be the same as before and 
  $\boldsymbol{v'}=(\alpha v_1',\dots,\alpha v_s',v_{s+1}',\dots,v_n')$, where $s=k-l-1\leq k-1$. We now consider the $e'$-Galois hull of the $[n+1,k]_{q}$ MDS code $\C=\GRS_k(\boldsymbol{a},\boldsymbol{v'},\infty)$. 
  
  A completely similar argument to 1) above and Part 2) of Theorem \ref{th.ConA} implies $\dim(\Hull_{e'}(\C))=l$, where $0\leq l\leq k-1$. This completes the proof.
  \end{proof}

  \begin{remark}\label{rem4}
  \begin{enumerate}
  \item [\rm 1)] According to Theorem \ref{th.ConB.1}, we can construct MDS codes with Galois hulls of arbitrary dimensions via Hermitian self-orthogonal (extended) GRS codes. 
  As we know, there are lots of Hermitian self-orthogonal (extended) GRS codes constructed in previous works.	
  \item [\rm 2)] Note that the conditions of Theorems \ref{th.ConA} and  \ref{th.ConB.1} are different, so they can lead to different results in general. We list some examples satisfying the 
  condition that $\frac{h}{\gcd(e',h)}$ is odd in Table \ref{tab:3}. Then the following example allows us to intuitively see the difference between Theorems \ref{th.ConA} and 
  \ref{th.ConB.1}. Start from the Hermitian (i.e., $5$-Galois) self-orthogonal (extended) GRS code over $\F_{p^{10}}$ with odd prime $p$. From Table \ref{tab:4}, we can see that 
  MDS codes with $0,\ 2,\ 4,\ 6$ and $8$-Galois hulls of arbitrary dimensions can be obtained by Theorem \ref{th.ConB.1}, while from Table \ref{tab:3}, we can see that only MDS codes with $5$-Galois 
  hulls of arbitrary dimensions can be derived from Theorem \ref{th.ConA}. Clearly, they are totally different.
  \end{enumerate}
  \end{remark}

  \begin{table}[!htb]
    \centering
    \caption{Some examples satisfying the condition that $\frac{h}{\gcd(e',h)}$ is odd}
    \label{tab:4}       
      \begin{tabular}{c|l|c|l}
        \hline
        $h$ & $e'$ & $h$ & $e'$ \\
        \hline
          $4$ & $0$ & $6$ & $0, 2, 4$ \\
          $8$ & $0$ & $10$ & $0, 2, 4, 6, 8$\\
          $12$ & $0, 4, 8$ & $14$ & $0, 2, 4, 6, 8, 10, 12$\\
        \hline
      \end{tabular}
    \end{table}

Note that in some known constructions of Hermitian self-orthogonal (extended) GRS codes, the dimension $k$ is usually roughly upper bounded by 
$\lfloor \frac{\sqrt{q}+n-1}{\sqrt{q}+1}\rfloor$ (e.g., see \cite[]{RefJ24,RefJ7,RefJ33,RefJ110,RefJ111} and references therein), while the dimension $k$ of our 
MDS codes with $e'$-Galois hulls of arbitrary dimensions constructed by Theorem \ref{th.ConB.1} is upper bounded by $\lfloor \frac{p^{e'}+n-1-\deg(h(x))}{p^{e'}+1}\rfloor$. 
We compare the magnitude of these two bounds in the following theorem, and it turns out that in some cases the range of dimension of the new MDS codes with $e'$-Galois hulls 
of arbitrary dimensions will be wilder.

\begin{theorem}\label{th.compare the magnitude}
  Let $q=p^h$ be a prime power, where $h$ is an even positive integer. Let $0\leq e'< \frac{h}{2}$ be an integer and $\deg(h(x))$ be known. Then 
  $\lfloor \frac{p^{e'}+n-1-\deg(h(x))}{p^{e'}+1}\rfloor > \lfloor \frac{\sqrt{q}+n-1}{\sqrt{q}+1}\rfloor$ if one of the following two conditions holds. 
  \begin{enumerate}
  \item [\rm 1)] $\deg(h(x))=0$ and $n\geq 3$; 
  \item [\rm 2)] $\deg(h(x))>0$, $0\leq e'\leq \lfloor \log_p\frac{\sqrt{q}(n-3)-(\sqrt{q}+1)\deg(h(x))-1}{\sqrt{q}+n-1} \rfloor$ and $n\geq \lfloor \frac{4\sqrt{q}+(\sqrt{q}+1)\deg(h(x))}{\sqrt{q}-1} \rfloor$.
  \end{enumerate}
\end{theorem}
  \begin{proof}
    For the condition 1), we note that when $\deg(h(x))=0$, the new upper bound is $\lfloor \frac{p^{e'}+n-1}{p^{e'}+1} \rfloor$. Let $k(e')=\frac{p^{e'}+n-1}{p^{e'}+1}$ be a function of $e'$, 
    where  $0\leq e'< \frac{h}{2}$. Easy to calculate, the first derivative of $k(e')$ is $$k'(e')=\frac{(2-n)p^{e'}lnp}{(p^{e'}+1)^2},\ 0\leq e'< \frac{h}{2}.$$
    It follows that $k(e')$ is monotonically decreasing for each $0\leq e'< \frac{h}{2}$ from the fact $k'(e')<0$ for $n\geq 3$. Since $0\leq e'< \frac{h}{2}$ is an integer, we 
    can easily conclude that the desired result holds.

   For the condition 2), let $\bigtriangleup k=\frac{p^{e'}+n-1-\deg(h(x))}{p^{e'}+1}- \frac{\sqrt{q}+n-1}{\sqrt{q}+1}$. Then it is easy to verify that under the condition 2), we have   
    \begin{align*}
      \bigtriangleup k= & \frac{(p^{e'}+n-1-\deg(h(x)))(\sqrt{q}+1)-(\sqrt{q}+n-1)(p^{e'}+1)}{(p^{e'}+1)(\sqrt{q}+1)} \\
                      = & \frac{(n-2)(\sqrt{q}-p^{e'})-(\sqrt{q}+1)\deg(h(x))}{(p^{e'}+1)(\sqrt{q}+1)} \\
                      \geq & 1.
    \end{align*} 
    It follows that $\lfloor \frac{p^{e'}+n-1-\deg(h(x))}{p^{e'}+1}\rfloor > \lfloor \frac{\sqrt{q}+n-1}{\sqrt{q}+1}\rfloor$. This completes the proof.
  \end{proof}

  \begin{remark}\label{rem5}
    \begin{enumerate}
      \item [1)] We consider condition 1) and condition 2) separately in Theorem \ref{th.compare the magnitude}. On one hand, condition 1) is more explicit and has a weaker requirement for $n$  
      than condition 2). Specifically, taking $\deg(h(x))=0$ in condition 2), it requires $n\geq \lfloor \frac{4\sqrt{q}}{\sqrt{q}-1} \rfloor=4+\lfloor \frac{4}{\sqrt{q}-1}\rfloor$, while $n\geq 3$ 
      is required in condition 1). 
      On the other hand, clearly, $\deg(h(x))=0$ is a relatively special situation. 

      \item [2)] According to Remarks \ref{rem.1} and \ref{rem.3}, we can easily determine $h(x)$ and $\deg(h(x))$ when a Hermitian self-orthogonal (extended) GRS code is given. 
      Hence, $\deg(h(x))$ in Theorem \ref{th.compare the magnitude} is indeed known.
    \end{enumerate}
  \end{remark}

  \begin{example}\label{example.thB.1}
    We list some known Hermitian self-orthogonal (extended) GRS codes and further explain the role of Theorem \ref{th.ConB.1}.
    \begin{enumerate}
    \item [1)] From Theorems 2.3, 2.5 of \cite{RefJ35} and Theorems 2, 3 of \cite{RefJ33}, there exists a $q$-ary Hermitian self-orthogonal GRS code of length $n$ if any of the following conditions holds:	
    \begin{itemize} 
    \item $t\mid (q-1)$, $r\leq \frac{q-1}{t}$, $k\leq \frac{t-1}{\sqrt{q}+1}$, $n=rt$ or $n=rt+1$;
    \item $2\leq n\leq q$, $n=n_1+n_2+\dots+n_t$ with $1\leq t\leq \sqrt{q}$ and $2\leq n_i\leq \sqrt{q}$ for all $1\leq i\leq t$, $1\leq k\leq \frac{\min\{n_1,\dots,n_t\}}{2}$;
    \item $1\leq m\leq \sqrt{q}$, $n=m(\sqrt{q}-1)$ and $1\leq k\leq \lfloor \frac{m\sqrt{q}-1}{\sqrt{q}+1} \rfloor$;
    \item $1\leq s\leq \sqrt{q}-1$, $n=s(\sqrt{q}+1)$ and $1\leq k\leq s-1$.
    \end{itemize}
    \item [2)] From Theorems 3.10, 3.11 of \cite{RefJ36}, Theorem 3.6 of \cite{RefJ38}, and Proposition 1 of \cite{RefJ24}, there exists a $q$-ary Hermitian self-dual or self-orthogonal extended GRS code of length $n+1$ if any of the following conditions holds:
    \begin{itemize} 
    \item  $n=2k-1$, $n\leq \sqrt{q}$, $\boldsymbol{a}=(a_1,a_2,\dots,a_n)\in A_l^n$, where the definition of $A_l^n$ see \cite{RefJ36};	
    \item  $n=2k-1$, $n\leq \sqrt{q}$, $\boldsymbol{a}=(a_1,a_2,\dots,a_n)\in A_{l,m}^n$, where the definition of $A_{l,m}^n$ see \cite{RefJ36};
    \item  $n=q$, $0\leq k\leq \sqrt{q}$;
    \item  $n=t(\sqrt{q}+1)+1$, $1\leq t\leq \sqrt{q}-1$ and $1\leq k\leq t+1$ with $(t,k)\neq (\sqrt{q}-1,\sqrt{q}-1)$.
    \end{itemize} 
   \end{enumerate}

   According to Theorem \ref{th.ConB.1} and Table \ref{tab:4}, if $h=6,\ 10,\ 14$, we can get twenty-four, forty and fifty-six classes of MDS codes with Galois hulls of arbitrary dimensions, respectively.
   \end{example}

  \begin{example}\label{example.thB.2}
    Similar to Example \ref{example.thA.1}, we show step by step how the new MDS codes with $e'$-Galois hulls of arbitrary dimensions are generated by known Hermitian self-orthogonal (extended) GRS codes.
    \begin{enumerate}
    \item [1)] New MDS codes with $e'$-Galois hulls of arbitrary dimensions can be derived from a Hermitian self-orthogonal GRS code as following steps: 
    
    $\bullet \textbf{Step 1. Obtain Hermitian self-orthogonal GRS codes.}$ 

    For example, taking $(q,m)=(3^6, 20)$, $h(x)=x^{7}\in \mathbb{F}_{3^6}[x]$ and $\lambda = 26\omega^{76}\in \mathbb{F}_{3^6}^*$ in Theorem 2 of \cite{RefJ33}, we can get $[520,k]_{3^6}$ 
    Hermitian self-orthogonal GRS codes for each $1\leq k\leq 19$. Consider the $[520,19]_{3^6}$ Hermitian self-orthogonal GRS code. 

    $\bullet \textbf{Step 2. Determine {\bf{deg}}(h(x)).}$ 

    Note that $\lambda u_ih(a_i)=v_i^{\sqrt{q}+1}$, $1\leq i\leq 520$ for $h(x)=x^7$ and $\lambda=26\omega^{76}$ in Theorem 2 of \cite{RefJ33}. Therefore, according to 
    Lemma \ref{coro.Hermitian self-orthogonal GRS and EGRS} and Theorem \ref{th.ConB.1}, $h(x)=x^7$ is the $h(x)$ we are looking for. Hence, $\deg(h(x))=7$. (Note that 
    we do not need to use the intermediate tool $g(x)$ at this time).

    $\bullet \textbf{Step 3. Derive new MDS codes with e'-Galois hulls of arbitrary dimensions.}$ 

    By the result 1) of Theorem \ref{th.ConB.1} and Table \ref{tab:4}, we can obtain MDS codes with $0,\ 2$ and $4$-Galois hulls of arbitrary dimensions as follows:  
    \begin{itemize}
    \item $[520, k_0]_{3^6}^{\textbf{**}}$ MDS code with $l$-dim $0$-Galois hull for $1\leq k_0\leq 256$, where $0\leq l\leq k_0$;	
    \item $[520, k_2]_{3^6}^{\textbf{**}}$ MDS code with $l$-dim $2$-Galois hull for $1\leq k_2\leq 52$, where $0\leq l\leq k_2$;
    \item $[520, k_4]_{3^6}$ MDS code with $l$-dim $4$-Galois hull for $1\leq k_4\leq 7$, where $0\leq l\leq k_4$.
    \end{itemize}

    \item [\rm 2)] New MDS codes with $e'$-Galois hulls of arbitrary dimensions can be derived from a Hermitian self-orthogonal extended GRS code as following steps: 
    
    $\bullet \textbf{Step 1. Obtain Hermitian self-orthogonal extended GRS codes.}$ 

    For example, taking $q=3^{10}$ and $(t, k)=(200, 100)$ in Proposition 1 of \cite{RefJ24}, we can get a $[48802,100]_{3^{10}}$ Hermitian self-orthogonal extended GRS code. 

    $\bullet \textbf{Step 2. Determine {\bf{deg}}(h(x)).}$ 

    Similar to Example \ref{example.thA.1}, since $h(x)=\lambda^{-1}g(x)$, we only need to determine the structure of $g(x)$. According to the proof of Proposition 1 of \cite{RefJ24}, 
    $g(x)=-m(x)^{\sqrt{q}+1}f^{\sqrt{q}}(x)$, where $m(x)\in \mathbb{F}_{3^{10}}[x]$ is a monic polynomial with $\deg(m(x))=t+1-k=101$. Taking $f(x)=1$ and $\lambda=-1$ further, 
    we have $h(x)=\lambda^{-1}g(x)=m(x)^{244}$. Hence, $\deg(h(x))=244\times 101=24644$.

    $\bullet \textbf{Step 3. Derive new MDS codes with e'-Galois hulls of arbitrary dimensions.}$ 

    By the result 2) of Theorem \ref{th.ConB.1} and Table \ref{tab:4}, we can obtain MDS codes with $0,\ 2,\ 4,\ 6$ and $8$-Galois hulls of arbitrary dimensions as follows:
  \begin{itemize}
  \item	$[48802, k_0]_{3^{10}}^{\textbf{**}}$ MDS code with $l$-dim $0$-Galois hull for $1\leq k_0\leq 12079$, where $0\leq l\leq k_0-1$;
  \item $[48802, k_2]_{3^{10}}^{\textbf{**}}$ MDS code with $l$-dim $2$-Galois hull for $1\leq k_2\leq 2416$, where $0\leq l\leq k_2-1$;
  \item $[48802, k_4]_{3^{10}}^{\textbf{**}}$ MDS code with $l$-dim $4$-Galois hull for $1\leq k_4\leq 295$, where $0\leq l\leq k_4-1$;
  \item $[48802, k_6]_{3^{10}}$ MDS code with $l$-dim $6$-Galois hull for $1\leq k_6\leq 34$, where $0\leq l\leq k_6-1$;
  \item $[48802, k_8]_{3^{10}}$ MDS code with $l$-dim $8$-Galois hull for $1\leq k_8\leq 4$, where $0\leq l\leq k_8-1$.
  \end{itemize}
  \end{enumerate}
\end{example}

\begin{remark}\label{rem6}

  We now discuss the dimensions of the new MDS codes with Galois hulls of arbitrary dimensions obtained in Example \ref{example.thB.2}. 
  Note that $\deg(h(x))=7$ in the result 1) of Example \ref{example.thB.2}. Hence, by Theorems \ref{th.ConB.1} and \ref{th.compare the magnitude}, we can obtain new MDS codes with 
  larger dimensions and $e'$-Galois hulls of arbitrary dimensions when $0\leq e'\leq \lfloor \log_p\frac{\sqrt{q}(n-3)-(\sqrt{q}+1)\deg(h(x))-1}{\sqrt{q}+n-1} \rfloor=
  \lfloor \log_3\frac{27\times(517)-(27+1)\times 7-1}{27+520-1} \rfloor=2$. For the result 2) of Example \ref{example.thB.2}, we note that $\deg(h(x))=24644$. Similarly, 
  we can obtain new MDS codes with larger dimensions and $e'$-Galois hulls of arbitrary dimensions when $0\leq e'\leq \lfloor \log_3\frac{243\times(48799)-(243+1)\times 24644-1}{243+48802-1} \rfloor=4$.  
  We have marked them in Example \ref{example.thB.2} with $\textbf{**}$.

\end{remark}
  
\subsection{Some explicit constructions of MDS codes with Galois hulls of arbitrary dimensions}\label{constructionC}

In this subsection, we construct two classes of MDS codes with Galois hulls of arbitrary dimensions. The conditions we use are relatively stringent. But applying Theorem \ref{th.ConA}, 
these newly obtained Galois self-orthogonal GRS codes will generate more MDS codes with Galois hulls of arbitrary dimensions. This means that, with the help of Theorem \ref{th.ConA}, many 
new MDS codes with Galois hulls of arbitrary dimensions can still be obtained by using stricter conditions. To this end, we need the following lemma.

\begin{lemma}\label{lem.ConC.1}
Let $m,h$ be two positive integers and $m\mid h$. Suppose that there is a positive integer $t$ such that $2^t\mid \frac{h}{m}$. Then for any element $x\in \mathbb{F}_{p^m}$, there exists $v\in \mathbb{F}_{p^h}$ such that $x=v^{2^t}$.
\end{lemma}
\begin{proof}
Since $2^t\mid \frac{h}{m}$, there exists a positive integer $s$ such that $h=2^t\cdot sm$. Then $\mathbb{F}_{p^h}$ can be viewed as the 
$2^t\cdot s$-th extension field of $\mathbb{F}_{p^m}$. It follows that any element $x\in \mathbb{F}_{p^m}$ can be written as $x=v^{2^t}$, where $v\in \mathbb{F}_{p^h}$.
\end{proof}

Let notations be the same as before, we now consider the additive subgroup of $\mathbb{F}_{p^h}$ and its cosets. We know that 
$\mathbb{F}_{p^h}$ can be seen as a linear space over $\mathbb{F}_{p^{m}}$ of dimension $\frac{h}{m}$. Suppose $1\leq w\leq p^{m}$ 
and $1\leq z\leq \frac{h}{m}-1$. Let $H$ be an $\mathbb{F}_{p^m}$-subspace of $\mathbb{F}_{p^h}$ of dimension $z$. 
Choose $\eta \in \mathbb{F}_{p^h}\setminus H$. We label the elements of $\mathbb{F}_{p^m}$ as $\beta_1=0,\beta_2,\cdots,\beta_{p^m}$. For $1\leq j\leq w$, define
\[\begin{split}
H_j=H+\beta_j\eta=\left\{h+\beta_j\eta|\ h\in H\right\}.
\end{split} \]
Since $\beta_i\neq \beta_j$ for any $1\leq i\neq j\leq n$, then $H_i\cap H_j=\emptyset$.
Let $n=wp^{mz}$ and
\begin{align}\label{eq.ConC_ai}
    \bigcup_{j=1}^{w}H_j=\left\{a_1,a_2,\dots,a_n\right\}.
\end{align}

Let $a_i$ and $u_i$ be defined as in Eqs. (\ref{eq.ConC_ai}) and (\ref{equation_ui}), we have the following lemma which has been shown 
in \cite[]{RefJ7}.

\begin{lemma}\label{lem.ConC.2}(\cite{RefJ7}, Lemma 3.1)
	For a given $1\leq i\leq n$, suppose $a_i\in H_b$ for some $1\leq b\leq w$. Then
\[\begin{split}
u_i=(\prod_{h\in H,h\neq0}h^{-1})(\prod_{g\in H}(\eta-g)^{1-w})(\prod_{1\leq j\leq t,j\neq b}(\beta_b-\beta_j)^{-1}).
\end{split} \]
In particular, let $\varepsilon=(\prod_{h\in H,h\neq0}h)(\prod_{g\in H}(\eta-g)^{w-1})$, then $\varepsilon u_i\in \mathbb{F}^*_{p^m}$.
\end{lemma}

The following lemma can be derived directly by Lemmas \ref{lem.ConC.1} and \ref{lem.ConC.2}.

\begin{lemma}\label{coro.ConC.1}
Let notations be the same as before, for any $\varepsilon u_i\in \mathbb{F}^*_{p^m}$, there exists $v_i\in \mathbb{F}_{p^h}^*$ such that 
\begin{align}\label{eq.ConC.ui}
    \varepsilon u_i=v_i^{2^t}.    
\end{align}
\end{lemma}

\begin{theorem}\label{th.ConC}
Let $p$ be an odd prime. Let $m, h$ be two positive integers with $m\mid h$ and $q=p^h$. Let $n=w p^{mz}$, where $1\leq w\leq p^m$ and $1\leq z\leq \frac{h}{m}-1$. Suppose that $t$ is a positive integer such that $2^t\mid \frac{h}{m}$ and $2^t=p^e+1$ for some $0\leq e\leq h-1$. Then the following hold.
\begin{enumerate}
\item [\rm 1)]For $1\leq k\leq \lfloor \frac{p^e+n-1}{p^e+1} \rfloor$, there exists an $[n,k]_q$ MDS code $\C$ with $\dim(\Hull_e(\C))=l$, where $0 \leq l\leq k$.	
\item [\rm 2)] For $1\leq k\leq \lfloor \frac{p^e+n-1}{p^e+1} \rfloor$, there exists an $[n+1,k]_q$ MDS code $\C$ with $\dim(\Hull_e(\C))=l$, where $0 \leq l\leq k-1$.
\end{enumerate}
\end{theorem}
\begin{proof}
1) Let notations be the same as before. By Lemma \ref{coro.ConC.1}, there exists $v_i\in \mathbb{F}_{p^h}^*$ such that $\varepsilon u_i=v_i^{2^t}$ for $1\leq i\leq n$. Since $(p^h-1)\nmid (p^e+1)$ for any $0\leq e\leq h-1$, there is an $\alpha \in \mathbb{F}_{p^h}^*$ such that $\beta=\alpha^{p^e+1}\neq 1$. Let $\boldsymbol{a}=(a_1,a_2,\dots,a_n)$ and $\boldsymbol{v}=(\alpha v_1, \dots, \alpha v_s,v_{s+1},\dots,v_n)$, where $s=k-l\leq k$. We now consider the $q$-ary MDS code $\C=\GRS_k(\boldsymbol{a},\boldsymbol{v})$. 

Note that $2^t=p^e+1$, similar to the proofs of the result 1) of Theorem \ref{th.ConA} and the result 1) of Theorem \ref{th.ConB.1}, we can prove that $\dim(\Hull_e(\C))=l$, where $0\leq l\leq k$.

2) Let $s=k-l-1\leq k-1$. We now consider the $q$-ary MDS code $\C=\GRS_k(\boldsymbol{a},\boldsymbol{v},\infty)$. Taking a similar way as the result 2) of Theorem \ref{th.ConA} and the result 2) of Theorem \ref{th.ConB.1}, we can also easily deduce that $\dim(\Hull_e(\C))=l$, where $0\leq l\leq k-1$. 
\end{proof}

\begin{table}[!htb]
  \centering
  \caption{Some examples satisfying the conditions $2^t=p^e+1$ and $2^t\mid \frac{h}{m}$ for $e=1$}
  \label{tab:5}       
  \begin{center}
    \begin{tabular}{ccccc|ccccc}
      \hline
      $p$ & $t$ & $m$ & $h$ & $q$ & $p$ & $t$ & $m$ & $h$ & $q$\\
      \hline
        3 & 2 & 1 & 4 & 81 & 3 & 2 & 2 & 8 & 6561 \\
        3 & 2 & 3 & 12 & 531441 & 3 & 2 & 4 & 16 & 43046721 \\
        3 & 2 & 5 & 20 & 3586784401 & 3 & 2 & 6  & 24 & 282429536481 \\
        3 & 2 & 7 & 28 & 22876792454961 & 3 & 2 & 8 &  32 & 1853020188851841 \\
        7 & 3 & 1 & 8 & 5764801 & 7 & 3 & 2 & 16 & 33232930569601 \\
        7 & 3 & 3 & 24 & 191581231380566414401 & 7 & 3 & 4 & 32 & 110442767424392064630529920 \\
      \hline
    \end{tabular}
  \end{center}
  \end{table}

\begin{example}\label{example.ConC.1}
In order to illustrate the practical significance of Theorem \ref{th.ConC}, we list some examples satisfying the condition $2^t=p^e+1$ and $2^t\mid \frac{h}{m}$ for $e=1$ in Table \ref{tab:5}. 
We explain that our constructions are new and flexible by the following comparisons.
\begin{enumerate}
\item [\rm 1)]From Table \ref{tab:1}, we can see that Classes $12,\ 13$ and $17$ have similar conditions with our results in Theorem \ref{th.ConC}. However, it is not difficult to find 
that $a\mid e$ (i.e., $m\mid e$) in Classes $12$ and $13$ or $t\mid (p^e-1)$ (i.e., $w\mid (p^e-1)$) in Class $17$ do not need to be satisfied in our constructions. Hence our constructions 
will get some codes with new lengths. For example, taking $(p,t,m,h,e)=(7,3,2,16,1)$ in Theorem \ref{th.ConC}, we have $n=w\cdot 49^z$, 
where $1\leq w\leq 49$, $1\leq z\leq 7$. Then applying Theorem \ref{th.ConC}, we still can get many MDS codes with $1$-Galois hulls of arbitrary dimensions, but Classes $12$ and $13$ will not 
be able to obtain these MDS codes because of $a\nmid e$ (since $a=2$, $e=1$ here, but $2\nmid 1$). And Class $17$ can only take $n=t\cdot7^r$, where $t=1,\ 2,\ 3,\ 6$ and $1\leq r\leq 15$. 
Cleaarly, most code lengths in Theorem \ref{th.ConC} are new.
\item [\rm 2)] Although we only list the case where $e=1$ in Table \ref{tab:5}, according to the result 1) of Theorem \ref{th.ConA}, we can get more new MDS codes with Galois hulls of arbitrary dimensions.
For example, taking $(p,t,m,h,e)=(3,2,1,8,1)$, we can obtain $[w\cdot3^z,k]_{3^8}$ $1$-Galois self-orthogonal GRS codes (i.e., taking $l=k$ in the result 1) of Theorem \ref{th.ConC}), where 
$1\leq w\leq 3$, $1\leq z\leq 7$ and $1\leq k\leq \lfloor \frac{w\cdot3^z+2}{4} \rfloor$. Similar to Example \ref{example.thA.1}, one can easily determine that $h(x)=1$ and $\deg(h(x))=0$. Then by 
the result 1) of Theorem \ref{th.ConA} and Table \ref{tab:3}, we can obtain $[w\cdot3^z,k_1]_{3^8}$ MDS codes with $3$-Galois hulls of arbitrary dimensions, $[w\cdot3^z,k_2]_{3^8}$ 
MDS codes with $5$-Galois hulls of arbitrary dimensions and $[w\cdot3^z,k_3]_{3^8}$ MDS codes with $7$-Galois hulls of arbitrary dimensions, where $1\leq k_1\leq \lfloor \frac{w\cdot3^z+26}{28} \rfloor$, 
$1\leq k_2\leq \lfloor \frac{w\cdot3^z+242}{244} \rfloor$ and $1\leq k_3\leq \lfloor \frac{w\cdot3^z+2186}{2188} \rfloor$. Finally, by Example \ref{example.thA.1} agian, these codes are also new. 
This fact shows that relatively strict conditions can also lead to many new classes of MDS codes with Galois hulls of arbitrary dimensions with the help of Theorem \ref{th.ConA}.
\end{enumerate}
\end{example}

\begin{remark}\label{rem7}
    We further discuss the condition $2^t=p^e+1$ when $e$ takes different values.
    \begin{enumerate}
    \item [\rm 1)] Take $e=0$, then $t\equiv 1$ and the results produced by the direct application of Theorem \ref{th.ConC} are exactly the conclusions 
    of Euclidean hull studied in the results (i) and (ii) of Theorem 3.3 of \cite[]{RefJ7}. Hence, Theorem \ref{th.ConC} is actually a generalization of \cite[]{RefJ7}.
    \item [\rm 2)] Take $e=1$, then $p=2^t-1$, which is the famous Mersenne prime. And $t$ must be a prime in this case. As we know, there is a well known 
    conjecture that the number of Mersenne primes is infinite. 
    \item [\rm 3)] Take $2\leq e\leq h-1$, then $p^e=2^t-1$, which implies that the odd prime $p$ should be the unique prime factor of a Mersenne composite number. 
    In fact, we can easily check that if $e$ is even, there is no $p$ satisfying the condition, but if $e$ is odd, this maybe an open problem because of the difficultity of decomposition 
    of Mersenne composite numbers.  
\end{enumerate}
For more information about Mersenne primes, we refer to \cite[]{RefJ52,RefJ51,RefJ50,RefJ53} and references therein.
\end{remark}

\section{Summary and concluding remarks}\label{sec4}
In this paper, three different methods are used to construct six new classes of MDS codes with Galois hulls of arbitrary dimensions (See Theorems \ref{th.ConA}, \ref{th.ConB.1} and \ref{th.ConC}). 
Specifically, the first two general methods allow us to construct more MDS codes with Galois hulls of arbitrary dimensions from Hermitian self-orthogonal (extended) GRS codes and general Galois 
self-orthogonal (extended) GRS codes. As stated in Theorem \ref{th.compare the magnitude}, the new MDS codes with Galois hulls of arbitrary dimensions derived from Hermitian self-orthogonal (extended) 
GRS codes have larger dimensions in some cases. 

In the third method, we use a relatively strict condition $2^t=p^e+1$ to present two explicit constructions of MDS codes with Galois hulls of arbitrary dimensions. In particular, one of them can derive 
$1$-Galois self-orthogonal GRS codes, and with the help of Theorem \ref{th.ConA}, more MDS codes with Galois hulls of arbitrary dimensions can be obtained from them directly (See Example \ref{example.ConC.1}). 
This fact shows that in our study, some relatively strict conditions can also lead to many new classes of MDS codes with Galois hulls of arbitrary dimensions. 

As one can see, in our constructions, the determination of the dimensions of the new MDS codes with Galois hulls of arbitrary dimensions derived from known Galois or Hermitian self-orthogonal (extended) GRS codes 
depends on $\deg(h(x))$. Fortunately, according to Remarks \ref{rem.1} and \ref{rem.3}, $\deg(h(x))$ is easy to be determined, whether it is calculated by Lagrange Interpolation Formula, or obtained directly from previous 
research results. 

In conclusion, the methods proposed in this paper are convenient and efficient in constructing MDS codes with Galois hulls of arbitrary dimensions. For future research, it might be interesting to construct more 
Galois self-orthogonal (extended) GRS codes.

\section*{Acknowledgments}
This research was supported by the National Natural Science Foundation of China (Nos.U21A20428 and 12171134).

\end{document}